\documentclass[a4paper,11pt]{article}
\usepackage{sobrapo-template}
\usepackage[english]{babel}
\usepackage[latin1]{inputenc}
\usepackage{amsmath,amssymb,amsthm}

\usepackage{indentfirst}

\newtheorem{lemma}{Lemma}[section]
\newtheorem{theorem}{Theorem}[section]
\newtheorem{corollary}{Corollary}[section]

\title{ON THE PACKING CHROMATIC NUMBER ON HAMMING GRAPHS AND GENERAL GRAPHS}

\begin{document}

\maketitle

\author{
\name{Graciela Nasini}
\institute{Facultad de Ciencias Exactas, Ingenier\'ia y Agrimensura - Universidad Nacional de Rosario\\CONICET}
\iaddress{Pellegrini 250, (2000) Rosario, Argentina}
\email{nasini@fceia.unr.edu.ar}
}

\author{ 
\name{Daniel Severin}
\institute{Facultad de Ciencias Exactas, Ingenier\'ia y Agrimensura - Universidad Nacional de Rosario\\CONICET}
\iaddress{Pellegrini 250, (2000) Rosario, Argentina}
\email{daniel@fceia.unr.edu.ar}
}

\author{ 
\name{Pablo Torres}
\institute{Facultad de Ciencias Exactas, Ingenier\'ia y Agrimensura - Universidad Nacional de Rosario\\CONICET}
\iaddress{Pellegrini 250, (2000) Rosario, Argentina}
\email{ptorres@fceia.unr.edu.ar}
}

\vspace{8mm}

\begin{abstract}
The packing chromatic number $\chi_\rho(G)$ of a graph $G$ is the smallest integer $k$ needed to proper color the vertices of $G$
in such a way the distance between any two vertices having color $i$ be at least $i+1$.
We obtain $\chi_\rho(H_{q,m})$ for $m=3$, where $H_{q,m}$ is the Hamming graph of words of length $m$ and alphabet with $q$ symbols,
and tabulate bounds of them for $m \geq 4$ up to 10000 vertices. We also give a polynomial reduction from the problem of finding
$\chi_\rho(G)$ to the Maximum Stable Set problem.
\end{abstract}

\bigskip

\newpage

\section{Introduction and previous results}

The packing chromatic number, originally called {\it broadcast chromatic number}, is a concept established by Goddard et al. (2008)
and comes from the area of frequency planning in wireless networks.
Given an undirected graph $G = (V,E)$, a {\it packing $k$-coloring} of $G$ is a partition $V_1,\ldots,V_k$ of $V$ such that for all
$i \in \{1, \ldots, k\}$, every distinct vertices $u,v \in V_i$ are at distance at least $i+1$. As usual, the smallest integer $k$ for which
exists a packing $k$-coloring of $G$ is called the {\it packing chromatic number} (PCN) of $G$ and it is denoted by
$\chi_\rho(G)$ (Bre\v{s}ar et al., 2007). Fiala and Golovach (2010) showed that determining this number is an $\mathcal{NP}$-hard problem even
for trees.

There are some interesting results relating this problem with the Maximum Stable Set (MSS) problem.
Let $G$ be a graph with $n$ vertices and diameter $d$. Goddard et al. (2008) proved that if $d=2$ then
$\chi_{\rho}(G)= n+1 - \alpha(G)$, where $\alpha(G)$ is the stability number of $G$ (i.e.$\!$ the cardinality of the maximum stable set).
Later, this result was generalized by Argiroffo et al. (2012) for graphs $G$ for which $\chi_{\rho}(G)\geq d$.

Let us introduce some notation. For any positive integer $k$ we denote $[k] = \{1, \ldots, k\}$.
From now on, we assume that a graph $G = (V, E)$ has $V = [n]$, it is connected and has diameter at least two.
Let $d_G(u,v)$ be the distance between $u$ and $v$ in $G$.
For $k \geq 2$, let $G^k$ be the graph such that $V(G^k) = V(G)$ and $E(G^k) = \{ (u,v) : d_G(u,v) \leq k \}$.
For $F \subset [n]$, let $G^F$ be the graph such that $V(G^F) = \{(v,k) : v \in V(G), k \in F \}$ and

$$E(G^F) = \bigcup_{k \in F} \{ ((u,k),(v,k)) : (u,v) \in E(G^k) \} \cup
                    \bigcup_{j,k \in F : j < k} \{ ((v,j),(v,k)) : v \in V(G) \}.$$
                    
Observe that, if $F=\{k\}$, then $G^F$ is isomorphic to $G^k$ and, if $G$ has diameter two, $G^{[d-1]}$ is isomorphic to $G$.

The next lemma includes the mentioned generalization on graphs with diameter two: 
\begin{lemma} (Argiroffo et al., 2012) \label{PROPPABLO}
Let $G=(V,E)$ be a graph with diameter $d$.
Then, $\chi_\rho(G) = \min \{ k : \alpha(G^{[k]}) = n \}$. Therefore, $\chi_\rho(G) \geq d$ if and only if $\alpha(G^{[d-1]}) < n$ and, in this case, $\chi_\rho(G) = (d-1) + n - \alpha(G^{[d-1]})$.
\end{lemma}

This result is useful for obtaining the exact value, or nontrivial bounds, of PCN for several families of graphs. In particular, Torres
(2013) applied it on the family of hypercubes. Hypercubes are particular cases of Hamming graphs, a main subject of study in coding theory. 

For any $q,m \in \mathbb{Z}^+$ with $q \geq 2$, the {\it Hamming graph} $H_{q,m}$ is the graph in which vertices
are represented as vectors $(u_1, \ldots, u_m) \in \{0,\ldots,q-1\}^m$, and two vertices are adjacent if and only if they differ in exactly
one coordinate (here, we use the definition given by Klotz and Sharifiyazdi (2008), but other authors, such as El Rouayheb et al. (2007),
uses a different definition).
Equivalently, $H_{q,m} = K_q \Box K_q \Box \ldots \Box K_q$ ($m$ times), where $K_q$ is the complete graph on $q$ vertices
and $G \Box H$ is the cartesian product of $G$ and $H$.
The Hamming graph $H_{q,m}$ has $q^m$ vertices, $m(q-1)q^m/2$ edges and diameter $m$. Finding the stability number of $H_{q,m}^k$, denoted by
$A_q(m,k+1)$ in coding theory literature (see, for instance, Huffman and Pless, 2003), is one of the most challenging problems in this area
and the exact value is known for just few cases, while for harder cases, only bounds are known.
As we shall see in the next section, the PCN of $H_{q,m}$ and the values $\alpha(H_{q,m}^k)$ are related each other.

Although exact values of PCN have been obtained for $m$-dimensional hypercubes up to $m=8$ (see Goddard et al., 2008, and Torres, 2013),
it is an open problem to know the PCN of remaining hypercubes, as well as the PCN of Hamming graphs for $q \geq 3$.\\
 
This work addresses two topics. In Section 2, we give a direct formula for computing $\chi_\rho(H_{q,m})$ for $m \in \{2,3\}$ and for all
$q$. We also get bounds for the PCN of Hamming graphs up to 10000 vertices by using an Integer Linear Programming approach.
In particular, we improve upper bounds for $m$-dimensional hypercubes with $m = 11, 12, 13$.

On the other hand, in Section 3, we present a polynomial reduction from the problem of computing the PCN of any graph to the MSS problem.
This reduction is given by the construction of a new graph $G_*^{[d-1]}$ from $G^{[d-1]}$, of size $(n+1)(d-1)$, and allows to address
the PCN for cases where previous results in the literature did not contemplate them.

\bigskip
\section{PCN of Hamming graphs} \label{hamm}

As we have mentioned in the introduction, Hamming graphs $H_{q,m}$ can be thought as the result of applying $m$ times the cartesian product
of the complete graph of $q$ vertices.
Bre\v{s}ar et al. (2007) proved that, for any graph $G$ with diameter $d$, $\chi_\rho(G \Box K_q) \geq q \chi_\rho(G) - (q-1) d$.
Then, we can obtain a lower bound of the PCN of $H_{q,m}$ as follows. 

\begin{lemma} \label{BRESAR}
$\chi_\rho(H_{q,m}) \geq m-1 + q^m - \sum_{k=1}^{m-1} q^{k}$.
\end{lemma}
\begin{proof}
The proof is by induction on $m$. If $m=1$, $\chi_\rho(H_{q,1}) = \chi_\rho(K_q) = q$.
For $m \geq 2$, the result given by Bre\v{s}ar et al. (2007) implies $\chi_\rho(H_{q,m+1}) \geq q \chi_\rho(H_{q,m}) - (q-1)m$.
Then, by the induction hypothesis,
$\chi_\rho(H_{q,m+1}) \geq  q (m-1 + q^m - \sum_{k=1}^{m-1} q^{k}) - (q-1)m = m + q^{m+1} - \sum_{k=1}^{m} q^{k}.$
\end{proof}

\medskip

Since $q^m - \sum_{k=1}^{m-1} q^{k}$ is a positive integer, $\chi_\rho(H_{q,m})$ is at less $m$, the diameter of $H_{q,m}$.
Therefore, we can apply Lemma \ref{PROPPABLO}, obtaining:

\begin{corollary} \label{COROLLARY2}
For $m \geq 2$, $\chi_\rho(H_{q,m}) = m-1 + q^m - \alpha(H_{q,m}^{[m-1]})$.
\end{corollary}

Now, we prove that the bound given in Lemma \ref{BRESAR} holds as equality for $H_{q,m}$ with
$m \in \{2,3\}$ and any $q\geq m$. 

\begin{theorem} \label{TEOREMITA}
(i) $\chi_\rho(H_{q,2}) = 1 + q^2 - q$.\\
(ii) If $q \geq 3$, $\chi_\rho(H_{q,3}) = 2 + q^3 - q^2 - q$.
\end{theorem}
\begin{proof}
In first place, let $\mathcal{S}_{q,m} = \{ (v_1, \ldots, v_m) \in V(H_{q,m}) : v_1 + \ldots + v_m = 0~(mod\ q)\}$ and
suppose that $u, v \in \mathcal{S}_{q,m}$. If $(u, v) \in E(H_{q,m})$, then there exists $j \in [m]$ such that $u_j \neq v_j$ and
$u_i = v_i$ for all $i \in [m] \backslash \{j\}$. However, $u_j = v_j ~(mod\ q)$ which is absurd. Therefore, $\mathcal{S}_{q,m}$
is a stable set of $H_{q,m}$ of size $q^{m-1}$.

Part (i). In virtue of Lemma \ref{BRESAR} and Corollary \ref{COROLLARY2}, we only need to prove $\alpha(H_{q,2}) \geq q$, which is true
since $|\mathcal{S}_{q,2}| = q$.

Part (ii). Let $H = H_{q,3}^{[2]}$. Again, we only need to prove
$\alpha(H) \geq q^2 + q$. We build a stable set $S$ of $H$ such that $|S| = q^2 + q$.

Define $a_0, a_1, \ldots, a_{q-1}$ as follows. For $i = 0,\ldots, \lfloor q/2 \rfloor-1$, $a_i = q - 2(i+1)$.
For $i = \lfloor q/2 \rfloor,\ldots,q-1$, $a_i = 2q - 2(i+1)$ if $q$ is odd, or $a_i = 2q - 2(i+1) + 1$ if $q$ is even.
It is not difficult to see that $(a_0, a_1, \ldots, a_{q-1}$) is a permutation of $(0, \ldots, q-1)$ and each $a_i$ satisfies $2i + a_i \neq 0~(mod\ q)$.

Now, consider $S = \{ (v,1) : v \in S^1 \} \cup \{ (v,2) : v \in S^2 \}$ where
$S^1 = \mathcal{S}_{q,3}$ and $S^2 = \{ (i,i,a_i) : i = 0,\ldots,q-1 \}$.
$S^1$ is a stable set of $H_{q,3}$ with size $q^2$, whereas $S^2$ is a stable set of $H_{q,3}^2$ with size $q$,
since vectors $(i,i,a_i)$ and $(j,j,a_j)$ differ in the 3 coordinates, for $i \neq j$.
Also, since each vertex $v=(v_1,v_2,v_3) \in S_2$ satisfies $v_1 + v_2 + v_3 = 2i + a_i \neq 0$ $(mod\ q)$, $S^1$ and $S^2$ are disjoint
sets, i.e.$\!$ if $u\in S_1$ then $(u,1)$ and $(v,2)$ are not adjacent in $H$.
\end{proof}

\medskip

Finally, we make some computational experiments in order to further our understanding of the PCN on Hamming graphs. By considering Corollary
\ref{COROLLARY2},
we focus on solving the MSS problem for $H_{q,m}^{[m-1]}$. We recall that state-of-the-art exact algorithms that solve the MSS problem report
an average size of instances less than 1000 vertices (see Rebennack, 2011).

We address this problem on Hamming graphs up to 10000 vertices, obtaining bounds in most of the cases.

For any $F \subset [m-1]$, the value $\alpha(H_{q,m}^F)$ can be obtained by solving the following $\{0,1\}$-linear programming
formulation of MSS problem:\\

\indent \indent \indent $max \sum_{k \in F} \sum_{v \in V(H_{q,m})} x_{v,k}$\\
\indent \indent \indent subject to\\
\indent \indent \indent \indent $\sum_{k \in F} x_{v,k} \leq 1 ~~~~~~~~~~~~\forall~v \in V(H_{q,m})$\\
\indent \indent \indent \indent $\sum_{v \in K} x_{v,k} \leq 1 ~~~~~~~~~~~~\forall~K \in \mathcal{K}_k$\\
\indent \indent \indent \indent $x_{v,k} \in \{0,1\} ~~~~~~~~~~~~~~~~\forall~v \in V(H_{q,m}), k \in F$\\

Here, stable sets of $H_{q,m}^F$ are described as the $\{0,1\}$-vectors $x_{v,k}$ and $\mathcal{K}_k$ is a particular polinomial
covering of edges by cliques of $H_{q,m}^k$. Such covering can be obtained as follows.
Define $N_G[v]$ as the closed the neighborhood of $v$ in $G$, i.e.$\!$ $N_G[v] = \{v\} \cup \{ u \in V(G) : (u,v) \in E(G)\}$.
If $k$ is even, consider $\mathcal{K}_k = \{ N_G[v] : v \in V \}$ where $G = H_{q,m}^{\frac{k}{2}}$.
The following straightforward result guarantees that this set is a suitable covering:
\begin{lemma}
Let $\mathcal{K} = \{ N_G[v] : v \in V \}$. For all $K \in \mathcal{K}$, $K$ is a clique of $G^2$. Moreover, for every $(u,v) \in E(G^2)$
there exists $K \in \mathcal{K}$ such that $u, v \in K$.
\end{lemma}
If $k$ is odd and $k\geq 3$, consider $\mathcal{K}_k = \{ N_{H_{q,m}^{\frac{k-1}{2}}}[v] : v \in V \} \cup J$ where $J$ has maximal
cliques containing the uncovered edges, obtained in a greedy form.

In order to find the packing chromatic number of $H_{q,m}$, we set $F = [m-1]$ and solve the resulting formulation through
IBM ILOG CPLEX 12.6. A time limit of 2 hours is imposed.

Since, in the harder cases, CPLEX is unable to obtain a reasonable initial incumbent, we also implement an heuristic
that finds a maximal stable set of $H_{q,m}^{[m-1]}$:
\begin{enumerate}
\item Compute $S \leftarrow \mathcal{S}_{q,m}$ (defined in the proof of Theorem \ref{TEOREMITA}).
\item For all $k = 2, \ldots, m-1$ do:
\begin{enumerate}
\item Find a maximal stable set $S'$ of $H_{q,m}^k[V(H_{q,m}^k) \backslash S]$ with a greedy heuristic.
\item Try to improve the size of $S'$ by solving the MSS formulation on $H_{q,m}^k$ with $x_{v,k} = 0$ for $v \in S$.
A time limit of 10 minutes is imposed to the solver.
\item Update $S \leftarrow S \cup S'$.
\end{enumerate}
\end{enumerate}

We also take advantage of Corollary \ref{COROLLARY2} and the result of Bre\v{s}ar et al. (2007) for getting an initial upper bound
of $\alpha(H_{q,m}^{[m-1]})$ in terms of the best lower bound of $\chi_\rho(H_{q,m-1})$ obtained in a previous stage:
$$\alpha(H_{q,m}^{[m-1]}) \leq m-1 + q^m - \bigl( q \chi_\rho(H_{q,m-1}) - (q-1) (m-1) \bigr)$$

The following table summarizes the resulting bounds after the optimization (a mark ``$-$'' means the Hamming graph has more than 10000):

\begin{center} \small
\begin{tabular}{|@{\hspace{3pt}}c@{\hspace{3pt}}|@{\hspace{3pt}}c@{\hspace{3pt}}c@{\hspace{3pt}}|@{\hspace{3pt}}c@{\hspace{3pt}}c@{\hspace{3pt}}|@{\hspace{3pt}}c@{\hspace{3pt}}c@{\hspace{3pt}}|@{\hspace{3pt}}c@{\hspace{3pt}}c@{\hspace{3pt}}|@{\hspace{3pt}}c@{\hspace{3pt}}c@{\hspace{3pt}}|@{\hspace{3pt}}c@{\hspace{3pt}}c@{\hspace{3pt}}|@{\hspace{3pt}}c@{\hspace{3pt}}c@{\hspace{3pt}}|@{\hspace{3pt}}c@{\hspace{3pt}}c@{\hspace{3pt}}|}
\hline
 q= & \multicolumn{2}{c}{3} & \multicolumn{2}{c}{4} & \multicolumn{2}{c}{5} & \multicolumn{2}{c}{6} & \multicolumn{2}{c}{7} & \multicolumn{2}{c}{8} & \multicolumn{2}{c}{9} & \multicolumn{2}{c|}{10} \\
  & LB & UB & LB & UB & LB & UB & LB & UB & LB & UB & LB & UB & LB & UB & LB & UB \\
\hline
$m=4$ & 48 & 48 & 175 & 175 & 473 & 473 & 1041 & 1043 & 2005 & 2009 & 3515 & 3525 & 5745 & 5757 & 8893 & 8908 \\
$m=5$ & 136 & 144 & 688 & 704 & 2349 & 2402 & 6226 & 6315 & $-$ & $-$ & $-$ & $-$ & $-$ & $-$ & $-$ & $-$ \\
$m=6$ & 399 & 432 & 2737 & 2908 & $-$ & $-$ & $-$ & $-$ & $-$ & $-$ & $-$ & $-$ & $-$ & $-$ & $-$ & $-$ \\
$m=7$ & 1185 & 1347 & $-$ & $-$ & $-$ & $-$ & $-$ & $-$ & $-$ & $-$ & $-$ & $-$ & $-$ & $-$ & $-$ & $-$ \\
$m=8$ & 3541 & 4099 & $-$ & $-$ & $-$ & $-$ & $-$ & $-$ & $-$ & $-$ & $-$ & $-$ & $-$ & $-$ & $-$ & $-$ \\
\hline
\end{tabular}
\end{center}

Note that the bound given by Lemma \ref{BRESAR} holds as equality for $m=4$ and $q \in\{4,5\}$, and the gap between
both bounds is certainly close for $q = 6,\ldots,9$. This observation makes us wonder whether, as in the case $m=3$, the bound given by
Lemma \ref{BRESAR} holds as equality for the case $m=4$ and $q \geq 4$.
In addition, the optimal stable sets found for $m=4$ and $q \in\{4,5\}$ have the form
$\{ (v,1) : v \in S^1 \} \cup \{ (v,2) : v \in S^2 \} \cup \{ (v,3) : v \in S^3 \}$ where $S^1 = \mathcal{S}_{q,4}$,
$S^2 = \{ (i,j,a_{ij}) : i,j = 0,\ldots,q-1\}$ and $S^3 = \{ (i,i,i,c_i) : i = 0,\ldots,q-1 \}$ where $a_{ij}$ is a certain
permutation of values $\{ (i,j) : i,j = 0,\ldots,q-1 \}$ and $c_i$ a certain permutation of $\{ 0,\ldots,q-1 \}$, so our conjecture could
be addressed by trying to construct such permutations.\\

On the other hand, our heuristic turned out to be helpful for hypercubes ($q=2$).
It allows us to improve the best known upper bound of PCN for $m=11$ given by Torres (2013), from value $881$ to $857$.
We also get values for $m \in \{12,13\}$: $\chi_\rho(H_{2,12}) \leq 1707$ and $\chi_\rho(H_{2,13}) \leq 3641$.

\bigskip
\section{A polynomial reduction from the PCN problem to the MSS problem} \label{polymss}

Let $G$ be a graph with $n$ vertices and diameter $d$. Observe that by Lemma \ref{PROPPABLO} we can obtain the value of $\chi_\rho(G)$ by
applying an iterative procedure. First, we compute $\alpha(G^{[d-1]})$.
If $\alpha(G^{[d-1]})< n$ we already know that $\chi_\rho(G) = (d-1) + n - \alpha(G^{[d-1]})$.
Otherwise, $\alpha(G^{[d-1]})=n$ and we can compute $\alpha(G^{[k]})$ with decreasing values of $k$ from $d-1$ to the first value $\hat{k}$ for which $\alpha(G^{[\hat{k}]})<n$. In this case, $\chi_\rho(G) = \hat{k}+1$.

In this section, we give a direct formula for obtaining $\chi_\rho(G)$ in terms of the stability number of a certain graph. 
Given $F \subset [n]$, we define the graph $G_*^F$ such that $V(G_*^F) = V(G^F) \cup \{(*,k) : k \in F \}$
and $E(G_*^F) = E(G^F) \cup \bigcup_{k \in F} \{ ((*,k),(v,k)) : v \in V(G) \}$. Given $S\subset V(G^F_*)$, we denote $K(S)= \{k \in F: (*,k) \in S \}$.
 
Clearly, $G^F_*$ can be constructed in polynomial time from $G$ for any $F\subset [n]$. Then, the mentioned reduction to the MSS problem is given by the following theorem:

\begin{theorem} \label{TEOREMAGEN}
Let $G$ be a connected graph with diameter $d \geq 2$.
Then,
$$\chi_\rho(G) = (d-1) + n - \alpha(G_*^{[d-1]}).$$
\end{theorem}

The proof of the previous theorem is based on the following properties of graphs $G^F_*$:

\begin{lemma} \label{LEMITA1}
(i)  Let $F \subset [n]$ and $S$ be a maximum stable set of $G_*^F$.
Then, there exists a maximum stable set $S^*$ of $G_*^F$ such that $K(S^*)$ is the set of $|K(S)|$ largest numbers of $F$. Moreover,
if $\alpha(G^F)< n$, there exists a maximum stable set $S$ of $G_*^F$ such that $K(S)=\varnothing$ and then, $\alpha(G_*^F)=\alpha(G^F)$. \\
(ii) Let $p\in [n]$. Then, $\alpha(G_*^{[p]})= max \{\alpha(G^{[t]})+p-t: t\leq p\}$. Moreover, if $\alpha(G^{[t]})=n$ for some $t\leq p$, $\chi_{\rho}(G)= n+p - \alpha(G_*^{[p]})$.
\end{lemma}

\begin{proof}
Part (i): If $k \in K(S)$, $S \cap V(G^k) = \varnothing$ and therefore $|S| = \alpha(G^{F \backslash K(S)}) + |K(S)|$.
Let $F'$ be the set of $|K(S)|$ largest numbers of $F$. 
Let $S'$ be a maximum stable set of $G^{F \backslash F'}$ and $S^* = S' \cup \{ (*,k) : k \in F'\}$. Clearly, $S^*$ is a stable set of $G_*^F$ such that $K(S^*)=F'$, with cardinality $\alpha(G^{F \backslash F'}) + |K(S)|$. Since 
$G^{F \backslash F'}$ is isomorphic to an edge subgraph of $G^{F \backslash K(S)}$, $\alpha(G^{F \backslash F'}) \geq \alpha(G^{F \backslash K(S)})$. Then, $|S^*|\geq |S|$, implying that $S^*$ is a maximum stable set of $G^F_*$.

Now, consider the case $\alpha(G^F) < n$. Let $S$ be a maximum stable set of $G_*^{F}$ and assume that there exists $k\in K(S)$.
Then, $S \cap V(G^k) = \varnothing$. Since $S\cap V(G^F)$ is a stable set of $G^{F}$ and $\alpha(G^{F})< n$, there exists $v\in V(G)$ such that $(v,j)\notin S$ for all $j\in F$. Then, $\tilde S= (S\setminus \{(*,k)\}) \cup \{(v,k)\}$ is a maximum stable set of $G_*^F$ such that $|K(\tilde S)|=|K(S)|-1$. The thesis follows by applying recursively this argument.

Part (ii): Let $\gamma= max \{\alpha(G^{[t]})+p-t: t\leq p\}$. Note that, for any $t\leq p-1$ and any stable set $S$ of $G^{[t]}$,
$S \cup \{(*,k): k = t+1, \ldots, p\}$ is a stable set of $G_*^{[p]}$ of size $|S|+ (p-t)$. Then, $\alpha(G_*^{[p]})\geq \gamma$.  On the other hand, from Part (i), there exists a maximum stable set $S^*$ of $G_*^{[p]}$ and some $t\leq p$ such that $S'=S^*\cap V(G^{[t]})$ is stable set of $G^{[t]}$ with cardinality $|S^*| - (p-t)$. Then, $\alpha (G^{[p]}_*)=|S^*| = |S'|+ p-t \leq \alpha(G^{[t]})+p-t\leq \gamma$.

Now, consider the case $\alpha(G^{[t]})=n$ for some $t\leq p$. Let $t' = \chi_{\rho}(G) = min \{t: \alpha(G^{[t]})=n \}$ by
Lemma \ref{PROPPABLO}. Clearly,  $t' \leq p $ and $\gamma= n + p - t'$. Therefore, $t' = n + p - \gamma = n + p - \alpha(G_*^{[p]})$.
\end{proof}

\begin{proof} ({\em of Theorem \ref{TEOREMAGEN}}).
Let $t = \chi_\rho(G)$. Then, $\alpha(G^{[t]})=n$.
If $t \leq d-1$, the result follows from Lemma \ref{LEMITA1} (ii). 
If $t \geq d$, by Lemma \ref{PROPPABLO}, $\alpha(G^{[d-1]}) < n$ and
$t = (d-1) + n - \alpha(G^{[d-1]})$, and the result follows from Lemma \ref{LEMITA1} (i), since $\alpha(G^{[d-1]}) = \alpha(G_*^{[d-1]})$.
\end{proof}

\medskip

Observe that, if we know an upper bound $t$ of the $\chi_\rho(G)$ such that $t < d-1$, since $\alpha(G_*^{[t]})=n$, Lemma \ref{LEMITA1} (ii) allows us to
compute $\chi_\rho(G)$ in terms of the stability number of $G_*^{[t]}$, which is smaller than $G_*^{[d-1]}$.

\medskip

\section*{Acknowledgments}

This work is supported by grants PID CONICET 11220120100277 and PICT-2013-0586.

\bigskip
\noindent{\bf References}

\noindent \textbf {Argiroffo G., Nasini G. and Torres P.} (2012),
The Packing Coloring Problem for $(q,q-4)$ Graphs,
\textit{LNCS} 7422, 309--319.

\noindent \textbf{Bre\v{s}ar B., Klav\v{z}ar S. and Rall D. F.} (2007),
On the packing chromatic numbers of cartesian products, hexagonal lattice, and trees,
\textit{Discrete Appl. Math.} 155, 2303--2311.

\noindent \textbf{Fiala J. and Golovach P. A.} (2010),
Complexity of the packing coloring problem of trees,
\textit{Discrete Appl. Math.} 158, 771--778.

\noindent \textbf{Goddard W., Hedetniemi S. M., Hedetniemi S. T., Harris J. and Rall D. F.} (2008),
Broadcast Chromatic Numbers of Graphs,
\textit{Ars Combinatoria} 86, 33--49.

\noindent \textbf{Huffman W. C. and Pless V.},
\textit{Fundamentals of Error-Correcting Codes},
Cambridge University Press, 2003.

\noindent \textbf{Klotz W. and Sharifiyazdi E.} (2008),
On the Distance Chromatic Number of Hamming Graphs,
\textit{Advances and Applications in Discrete Mathematics} 2, 103--115.

\noindent \textbf{Rebennack S., Oswald M., Theis D. O., Seitz H., Reinelt G. and Pardalos P. M.} (2011),
A Branch and Cut solver for the maximum stable set problem,
\textit{J. Comb. Optim.} 21, 434--457.

\noindent \textbf{El Rouayheb S. Y., Georghiades C. N., Soljanin E., Sprintson, A.} (2007),
Bounds on Codes Based on Graph Theory,
\textit{Proc. IEEE International Symposium on Information Theory}, 1876--1879.

\noindent \textbf{Torres P. and Valencia-Pabon M.} (2013),
On the packing chromatic number of hypercubes,
\textit{ENDM} 44, 263--268.

\end{document}